\documentclass[10pt, conference]{IEEEtran}
\usepackage{amsthm}
\usepackage{amsfonts}
\usepackage{graphicx}
\usepackage{latexsym}
\usepackage{amssymb}
\usepackage{stmaryrd}
\usepackage{comment}
\usepackage{amscd}
\usepackage[small]{caption}
\usepackage[belowskip=-10pt,aboveskip=0pt]{caption}
\usepackage[dvipsnames]{xcolor}
\usepackage{algorithm}
\usepackage[noend]{algpseudocode}
\algrenewcommand\alglinenumber[1]{\scriptsize #1:}
\algrenewcommand\algorithmicindent{1em}%
\usepackage{graphicx}
\usepackage{subfigure}
\usepackage{wrapfig}
\usepackage{float}
\usepackage{bbm}
\usepackage{amsmath,bm}
\usepackage{thmtools} 
\usepackage{thm-restate}
\usepackage{subcaption}

\usepackage{tikz}
\tikzset{
    cross/.pic = {
    \draw[rotate = 45] (-#1,0) -- (#1,0);
    \draw[rotate = 45] (0,-#1) -- (0, #1);
    }
}
\usetikzlibrary{patterns,arrows}
\usetikzlibrary{decorations.pathreplacing,calligraphy}
\usepackage{graphicx} % Allows including images
\usepackage{booktabs} % Allows the use of \toprule, \midrule and \bottomrule in tables
\allowdisplaybreaks
\usepackage{mathtools}
\newcommand{\bea}{\begin{eqnarray}}
\newcommand{\eea}{\end{eqnarray}}
\newcommand{\bean}{\begin{eqnarray*}}
\newcommand{\eean}{\end{eqnarray*}}

\newcommand{\sbinom}[2]{\left( \begin{array}{c} #1 \\ #2 \end{array} \right) }

\newcommand{\tx}{\Tilde{\bfx}}

\newcommand{\cA}{{\cal A}}
\newcommand{\cB}{{\cal B}}

\newcommand{\cE}{{\cal E}}

\newcommand{\cG}{{\cal G}}

\newcommand{\cL}{{\cal L}}
\newcommand{\cM}{{\cal M}}
\newcommand{\cN}{{\cal N}}

\newcommand{\cP}{{\cal P}}

\newcommand{\cR}{{\cal R}}
\newcommand{\cS}{{\cal S}}
\newcommand{\cT}{{\cal T}}

\newcommand{\cX}{{\cal X}}
\newcommand{\cY}{{\cal Y}}

\newcommand{\sG}{\script{G}}

\newcommand{\sP}{\script{P}}

\newcommand{\bfa}{{\boldsymbol a}}
\newcommand{\bfb}{{\boldsymbol b}}

\newcommand{\bfd}{{\boldsymbol d}}

\newcommand{\bfx}{{\boldsymbol x}}
\newcommand{\bfy}{{\boldsymbol y}}

\DeclareMathOperator*{\argmin}{arg\,min}

\DeclareMathAlphabet{\mathbfsl}{OT1}{cmr}{bx}{it}
\newcommand{\uuu}{\kern-1pt\mathbfsl{u}\kern-0.5pt}
\newcommand{\vvv}{\kern-1pt\mathbfsl{v}\kern-0.5pt}

\newcommand{\myboxplus}{\kern1pt\mbox{\small$\boxplus$}}

\newcommand\blfootnote[1]{%
  \begingroup
  \renewcommand\thefootnote{}\footnote{#1}%
  \addtocounter{footnote}{-1}%
  \endgroup
}

\makeatletter \DeclareRobustCommand{\sbinom}{\genfrac[]\z@{}}
\makeatother
\newcommand{\G}[2]{\sbinom{{#1}\kern-1pt}{{#2}\kern-1pt}}
\newcommand{\Gq}[2]{\sbinom{{#1}\kern-0.25pt}{{#2}\kern-0.25pt}}

\newcommand{\Ps}{\smash{{\sP\kern-2.0pt}_q\kern-0.5pt(n)}}
\newcommand{\sPs}{\smash{{\sP\kern-1.5pt}_q(n)}}
\newcommand{\Ptwo}{\smash{{\sP\kern-2.0pt}_2\kern-0.5pt(n)}}
\newcommand{\Ptwom}{\smash{{\sP\kern-2.0pt}_2\kern-0.5pt(m)}}
\newcommand{\Ptwonm}{\smash{{\sP\kern-2.0pt}_2\kern-0.5pt(n+m)}}
\newcommand{\Ptwoa}{\smash{{\sP\kern-2.0pt}_2\kern-0.5pt(1)}}
\newcommand{\Ptwob}{\smash{{\sP\kern-2.0pt}_2\kern-0.5pt(2)}}
\newcommand{\Ptwoc}{\smash{{\sP\kern-2.0pt}_2\kern-0.5pt(3)}}
\newcommand{\Ptwod}{\smash{{\sP\kern-2.0pt}_2\kern-0.5pt(4)}}
\newcommand{\Ptwoe}{\smash{{\sP\kern-2.0pt}_2\kern-0.5pt(5)}}
\newcommand{\Ptwof}{\smash{{\sP\kern-2.0pt}_2\kern-0.5pt(6)}}
\newcommand{\Ptwokm}{\smash{{\sP\kern-2.0pt}_2\kern-0.5pt(2k-1)}}
\newcommand{\Pone}{\smash{{\sP\kern-2.5pt}_2\kern-0.5pt(n{-}1)}}

\newcommand{\Gr}{\smash{{\sG\kern-1.5pt}_q\kern-0.5pt(n,k)}}
\newcommand{\Gi}{\smash{{\sG\kern-1.5pt}_q\kern-0.5pt(n,i)}}
\newcommand{\Gj}{\smash{{\sG\kern-1.5pt}_q\kern-0.5pt(n,j)}}
\newcommand{\Grmk}{\smash{{\sG\kern-1.5pt}_q\kern-0.5pt(n,n-k)}}
\newcommand{\Grdk}{\smash{{\sG\kern-1.5pt}_q\kern-0.5pt(2k,k)}}
\newcommand{\Grekappa}{\smash{{\sG\kern-1.5pt}_q\kern-0.5pt(n,e+1-\kappa)}}
\newcommand{\Grtwoekappa}{\smash{{\sG\kern-1.5pt}_q\kern-0.5pt(n,2e+1-\kappa)}}
\newcommand{\Gremkappa}{\smash{{\sG\kern-1.5pt}_q\kern-0.5pt(n,e-\kappa)}}
\newcommand{\Gn}{\smash{{\sG\kern-1.5pt}_2\kern-0.5pt(n,n{-}1)}}
\newcommand{\Gnq}{\smash{{\sG\kern-1.5pt}_q\kern-0.5pt(n,n{-}1)}}
\newcommand{\Gone}{\smash{{\sG\kern-1.5pt}_2\kern-0.5pt(n,1)}}
\newcommand{\Gqone}{\smash{{\sG\kern-1.5pt}_q\kern-0.5pt(n,1)}}
\newcommand{\GTwo}{\smash{{\sG\kern-1.5pt}_2\kern-0.5pt(n,k)}}
\newcommand{\GTwonk}[2]{{\smash{{\sG\kern-1.5pt}_2\kern-0.5pt({#1},{#2})}}}
\newcommand{\Gnk}{\smash{{\sG\kern-1.5pt}_2\kern-0.5pt(n,n{-}k)}}
\newcommand{\Greone}{\smash{{\sG\kern-1.5pt}_q\kern-0.5pt(n,e{+}1)}}
\newcommand{\Gretwo}{\smash{{\sG\kern-1.5pt}_q\kern-0.5pt(n,e{+}2)}}

\newcommand{\be}[1]{\begin{equation}\label{#1}}
\newcommand{\ee}{\end{equation}}

\newcommand{\BDC}{\mathsf{BDC}}

\newcommand{\Cref}[1]{Co\-rol\-la\-ry\,\ref{#1}}

\newtheorem{lemma}{Lemma}

\newtheorem{corollary}{Corollary}

\newtheorem{definition}{Definition}
\newtheorem{proposition}{Proposition}

\newtheorem{problem}{Problem}

\newcommand{\SCS}{\mathsf{SCS}}
\newcommand{\LCS}{\mathsf{LCS}}



\IEEEoverridecommandlockouts

\begin{document}

\author{\IEEEauthorblockN{ \textbf{Shubhransh~Singhvi}$^{1*}$, \textbf{Charchit~Gupta}$^{2*}$, \textbf{Avital Boruchovsky}$^3$, \textbf{Yuval Goldberg}$^3$,\\ 
\textbf{Han~Mao~Kiah}$^4$ and \textbf{Eitan~Yaakobi}$^3$}
  \IEEEauthorblockA{$^1$% 
  Signal Processing  \&  Communications Research  Center, IIIT Hyderabad, India}
\IEEEauthorblockA{$^2$%
                      IIIT Hyderabad, India}
\IEEEauthorblockA{$^3$%
                     Department of Computer Science, %\\
                     Technion---Israel Institute of Technology, 
                     Haifa 3200003, Israel}
  \IEEEauthorblockA{$^4$%
                     School of Physical and Mathematical Sciences, 
		Nanyang Technological University, Singapore}
 }
 
\title{\textbf{Permutation Recovery Problem against Deletion Errors for DNA Data Storage }}
% \author{
% \IEEEauthorblockN{Anonymous Authors}
% }
\date{\today}
\maketitle
\thispagestyle{empty}	
\pagestyle{empty}
%%%%%%%%

%%%%%%%%

\hspace*{-3mm}\begin{abstract}
Owing to its immense storage density and durability, DNA has emerged as a promising storage medium.  However, due to technological
constraints, data can only be written onto many short DNA molecules called \textit{data blocks} that are stored in an unordered way. To handle the unordered nature of DNA data storage systems, a unique \textit{address} is typically prepended to each data block to form a \textit{DNA strand}. However, DNA storage systems are prone to errors and generate multiple noisy copies of each strand called \textit{DNA reads}. Thus, we study the \emph{permutation recovery problem} against deletions errors for DNA data storage.
    
The permutation recovery problem for DNA data storage requires one to reconstruct the addresses or in other words to uniquely identify the noisy reads. By successfully reconstructing the addresses, one can essentially determine the correct order of the data blocks, effectively solving the clustering problem.

We first show that we can almost surely identify all the noisy reads under certain mild assumptions. We then propose a permutation recovery procedure and analyze its complexity.  

\end{abstract}
\blfootnote{$^*$Equal contribution
% The first two authors contributed equally to this work.
}
\section{Introduction}
The need for a more durable and compact storage system has become increasingly evident with the explosion of data in modern times. While magnetic and optical disks have been the primary solutions for storing large amounts of data, they still face limitations in terms of storage density and physical space requirements. Storing a zettabyte of data using these traditional technologies would necessitate a vast number of units and considerable physical space.

The idea of using macromolecules for ultra-dense storage systems was recognized as early as the 1960s when the physicist Richard Feynman outlined his vision for nanotechnology in his talk `There is plenty of room at the bottom'. Using DNA is an attractive possibility because it is extremely dense (up to about 1 exabyte per cubic millimeter) and durable (half-life of over 500 years).  Since the first experiments conducted by Church et al. in 2012~\cite{Church.etal:2012} and Goldman et al. in 2013~\cite{Goldman.etal:2013}, there have been a flurry of experimental demonstrations (see \cite{Shomorony.2022,Yazdi.etal:2015b} for a
survey). Amongst the various coding design considerations, in this work, we study the unsorted nature of the DNA storage system~\cite{LSWY18, Shomorony.2022}.

A DNA storage system consists of three important components. The first is the DNA synthesis which produces the oligonucleotides, also called \textit{strands}, that encode the data. The second part is a storage container with compartments  which  stores the  DNA  strands, however without order. Finally, to retrieve the data, the DNA is accessed using next-generation sequencing, which results in several noisy copies, called \textit{reads}. The processes of synthesizing, storing, sequencing, and handling strands are all error prone. Due to this unordered nature of DNA-based storage systems, when the user retrieves the information, in addition to decoding the data, the user has to determine the identity of the data stored in each strand. A typical solution is to simply have a set of addresses and store this address information as a prefix to each DNA strand. As the addresses are also known to the user, the user can identify the information after the decoding process. As these addresses along with the stored data are prone to errors, this solution needs further refinements.

In \cite{Organick}, the strands (strand = address + data) are first clustered with respect to the edit distance. Then the authors determine a consensus output amongst the strands in each cluster and  finally, decode these consensus outputs using a classic concatenation scheme.  
For this approach,  the clustering step is computationally expensive.
When there are $\cM$ reads, the usual clustering method involves $\cM^2$ pairwise comparisons to compute distances. This is costly when the data strands are long, and 
the problem is further exacerbated if the metric is the edit distance.
Therefore, in  \cite{clustering}, a distributed approximate clustering algorithm was proposed and the authors clustered 5 billion strands in 46 minutes on 24 processors.

In \cite{CKVY23}, the authors proposed and investigated an approach that avoids clustering, by studying a generalisation of the \textit{bee identification problem}. Informally, the bee identification problem requires the receiver to identify $M$ “bees” using a set of $M$ unordered noisy measurements~\cite{TTV2019}. 
Later, in~\cite{CKVY23}, the authors generalized the setup to multi-draw channels where every bee (address) results in $N$ noisy outputs (noisy addresses). The task then is to identify each of the $M$ bees from the $MN$ noisy outputs and it turns out that this task can be reduced to a minimum-cost network flow problem. In contrast to previous works, the approach in~\cite{CKVY23} utilizes only the address information, which is of significantly shorter length, and the method does not take into account the associated noisy data. Hence, this approach involves no data comparisons.

However, as evident, the clustering and bee identification based approaches do not completely take into account the nature of the DNA storage system. In particular, the clustering approaches do not utilise the uncorrupted set of addresses which can be accessed by the receiver and the bee identification approach uses solely the information stored in address and neglects the noisy data strands.   

In \cite{SBKY23}, the authors devised an approach that utilizes both the address and data information to identify the noisy reads. However, their approach was designed for the binary erasure channel, i.e., when the reads are corrupted by erasures.  

In this paper, we consider the more challenging noise model of deletions; a more realistic noise model for DNA data storage. Specifically, for the binary deletion channel, we first show that we can almost surely correctly identify all the reads under certain mild assumptions. Then we propose our permutation recovery procedure and demonstrate that on average the procedure uses only a fraction of $\cM^2$ data comparisons (when there are $\cM$ reads).

\section {Problem Formulation}
\label{sec:prob}
Let $N$ and $M$ be positive integers. Let $[M]$ denote the set $\{1,2,\ldots,M\}$. An $N$-permutation $\pi$ over $[M]$ is an $NM$-tuple $(\pi(j))_{j\in[MN]}$ where every symbol in $[M]$ appears exactly $N$ times, and we denote the set of all $N$-permutations over $[M]$ by $\mathbb{S}_N(M)$.
Let the set of addresses be denoted by $\cA \subseteq \{0,1\}^n$ and $M \triangleq |\cA|$. We will use the terms addresses and codewords interchangeably. We assume that every codeword $\bfx_i \in \cA$ is attached to a length-$L$ data part $\bfd_i\in {\{0,1\}}^L$ to form a strand, which is the tuple, $(\bfx_i,\bfd_i)$. Let the multiset of  data be denoted by $D = \{\{\bfd_i:i\in[M]\}\}$ and the set of strands by $R=\{(\bfx_i,\bfd_i):i\in[M]\}$. Throughout this paper, we assume that $D$ is drawn uniformly at random over $\{0,1\}^L$. Let $\cS_N((\bfx,\bfd))$ denote the multiset of channel outputs when $(\bfx,\bfd)$ is transmitted $N$ times through the channel $\cS$. Assume that the entire set $R$ is transmitted through the channel $\cS$, hence an unordered multiset, $R'=\{\{(\bfx_1',\bfd_1'),(\bfx_2',\bfd_2'),\ldots,(\bfx'_{MN},\bfd_{MN}')\}\}$, of $MN$ noisy strands (reads)  is obtained, where for every $j\in[MN]$, $(\bfx_j', \bfd_j') \in \cS_N((\bfx_{\pi(j)},\bfd_{\pi(j)}))$ for some  $N$-permutation $\pi$ over $[M]$, which will be referred to as the \textit{true $N$-permutation}. Note that the receiver, apart from the set of reads $R'$, has access to the set of addresses $\cA$ but does not know the set of data $D$. Let $\Delta > 0$ and $L = \Delta n$. 

For an integer $k$, $0\le k\le n$, a sequence $\bfy\in\{0,1\}^{n-k}$ is a  \emph{$k$-subsequence} of $\bfx\in\{0,1\}^n$ if $\bfy$ can be obtained by deleting $k$ symbols from $\bfx$. Similarly, a sequence $\bfy\in\{0,1\}^{n+k}$ is a  \emph{$k$-supersequence} of $\bfx\in\{0,1\}^n$ if $\bfx$ is a $k$-subsequence of~$\bfy$. Let $\bfx$ and $\bfy$ be two sequences of length $n$ and $m$ respectively such that $m < n$. The \emph{embedding number} of $\bfy$ in $\bfx$, denoted by $\omega_{\bfy}(\bfx)$, is defined as the number of distinct occurrences of $\bfy$ as a subsequence of $\bfx$. More formally, the embedding number is the number of distinct index sets, $(i_1, i_2, \ldots, i_{m})$, such that $1 \le i_1 < i_2 < \cdots < i_{m} \le n$ and $x_{i_1} = y_1, x_{i_2} = y_2, \dotsc, x_{i_{m}} = y_{m}$. For example, for $\bfx = \texttt{11220}$ and $\bfy = \texttt{120}$, it holds $\omega_{\bfy}(\bfx) = 4$. The \textit{$k$-insertion ball} centred at ${\bfx\in\{0,1\}^n}$, denoted by $I_k(\bfx)\subseteq \{0,1\}^{n+k}$, is the set of all $k$-supersequences of $\bfx$. Similarly, the {\textit{$k$-deletion ball}} centred at ${\bfx\in\{0,1\}^n}$, denoted by $D_k(\bfx)\subseteq \{0,1\}^{n-k}$, is the set of all $k$-subsequences of~$\bfx$. Let $\bfx, \bfy \in \{0,1\}^*$, we denote the shortest common supersequence between $\bfx,\bfy$ by $\SCS(\bfx,\bfy)$ and the longest common subsequence by $\LCS(\bfx,\bfy)$. 

% For $m,L\in \mathbb{Z}_{+}$ and $\bfx \in \{0,1\}^n$, let $\cN(\bfx,m,L) \triangleq \{\bfy : \bfy \in \{0,1\}^m, \LCS(\bfx,\bfy) = L\}$. A trivial way to compute $\cN(\bfx,n,L)$ is to go through all the words in $\{0,1\}^n$ and compute the $\LCS$, which takes $O(n^22^{n})$ time. In the next lemma, we give a recursive expression to compute $\cN(\bfx,n,L)$ in $O(n^3)$ time. 

% \begin{lemma}
% For $\bfx \in \{0,1\}^n$,
% \begin{align*}
%     \cN(\bfx,n,L) = \cN(\bfx,n-1,L) \circ \overline{\bfx}_{[n]} +  \cN(\bfx_{[1:n-1]},n-1,L-1) \circ \bfx_{[n]}.
% \end{align*}
% \end{lemma}
% \begin{proof}
% We prove the lemma using the principle of induction. For the base case, i.e., $n = L = 1$,  it can be verified that$\cN(\bfx,n,L) = \{\bfx\}$ = . Consider now some entry $\cN(\bfx,n,L)$, and assume that the claim holds for $\cN(\bfx,n-1,L)$ and $\cN(\bfx_{[1:n-1]},n-1,L-1)$. Considering the last concatenated bit, the terms $\cN(\bfx,n-1,L) \circ \overline{\bfx}_{[n]}$ and $\cN(\bfx_{[1:n-1]},n-1,L-1) \circ \bfx_{[n]}$ would have $L + 0$ and $(L-1) + 1$ as their lcs. The unequal last bits also ensure that there's nothing common between the two terms.\\
% \end{proof}

In \cite{SBKY23}, the authors introduced the following problems:

\begin{problem}
\label{prob1}
Given $\cS, \epsilon$ and $\cA$, find the region $\cR\in \mathbb{Z}_{+}^2$, such that for $(N,L) \in \cR$,  it is possible to identify the true permutation with probability at least $1-\epsilon$ when the data $D$ is drawn uniformly at random. 
\end{problem}

% For $(N,L)\in\cR$, we can find the true permutation by making at least $(NM)^2$ data comparisons. This may be expensive when the data parts are long, i.e., when $L$ is large. Therefore, our second objective is to reduce the number of data comparisons.

\begin{problem}
\label{prob2}
   Let $\kappa<1$. Given $\cS, \epsilon, \cA$ and $(N,L)\in \cR$, design an algorithm to identify the true permutation with probability at least $1-\epsilon$ using $\kappa (NM)^2$ data comparisons. As before, $D$ is  drawn uniformly at random. 
\end{problem}

In Section~\ref{sec:multi-draw}, we demonstrate that the algorithm in \cite{CKVY23} identifies the true permutation with a vanishing error probability as $n$ grows. In Section~\ref{sec:uniq_n}, we address Problem~\ref{prob1} and identify the region $\cR$ for which there exists only one valid permutation, viz. the true permutation. In Section~\ref{sec:perm_alg}, we describe our algorithm that identifies the true-permutation with probability at least $1-\epsilon$ when $(N,L) \in \cR$ and also analyse the expected number of data comparisons performed by the algorithm. Unless otherwise mentioned, we consider the channel to be the binary deletion channel, i.e., $\cS =\mathsf{BDC}(p)$, where $p$ denotes the deletion probability.

Due to space limitations, we give some of the proofs in the appendix.

\section{Bee-Identification over Multi-Draw Deletion Channels}
\label{sec:multi-draw}
The algorithm in \cite{CKVY23} uses solely the information stored in the addresses to identify the true-permutation, and does not take into consideration the noisy data that is also available to the receiver.  The first step in their algorithm is to construct a bipartite graph $\cG = (\cX\cup \cY, E)$ as follows:

\begin{enumerate}
    \item Nodes: the left nodes are the addresses ($\cX = \cA$) and the right nodes are the noisy reads ($\cY = R'$).
    \item Demands: for each left node $\bfx\in \cX$, we assign a demand $\boldsymbol{\delta}(\bfx) =-N$, while for each right node $(\bfy,\bfd')\in \cY$, we assign a demand $\boldsymbol{\delta}(\bfy) = 1$.
    \item Edges: there exists an edge between $\bfx\in \cX$ and $(\bfy,\bfd')\in \cY$ in $E$ if and only if $P(\bfy|\bfx)>0$, where $P(\bfy|\bfx)$ is the likelihood probability of observing $\bfy$ given that $\bfx$ was transmitted. 
    \item  Costs: for an edge $(\bfx, (\bfy,\bfd')) \in E$, we assign the cost $\gamma((\bfx, (\bfy,\bfd'))) =-\log P(\bfy \mid \bfx)$.
\end{enumerate}

For $\bfx\in \cX$ and $(\bfy,\bfd')\in\cY$, let $E_\bfx$ and $E_{(\bfy,\bfd')}$ denote the multiset of neighbours of $\bfx$ and the set of neighbours of $(\bfy,\bfd')$ in $\cG$, respectively, i.e.,  $E_\bfx=\{\{(\bfy,\bfd')|(\bfx,(\bfy,\bfd'))\in E\}\}$, $E_{(\bfy,\bfd')}=\{\bfx|(\bfx,(\bfy,\bfd'))\in E\}$. Note that the degree of every left node is at least $N$ as $\cS_N((\bfx,\bfd)) \subseteq{} E_\bfx$.

The likelihood of an $N$-permutation $\pi$ can be computed to be $ \prod_{(\bfy,\bfd') \in \cY} P(\bfy | \pi(\bfy))$. Further, it was shown that the task of finding the permutation that maximizes this probability can be reduced to the task of finding a minimum-cost matching on the graph $\cG$. 

However, since this approach neglects the information stored in the data, it is not necessary for the permutation recovered by the minimum-cost matching to be the true permutation. When the set of addresses is the entire space, i.e., $M = 2^n$, we show in the following lemma that this algorithm identifies the true permutation with a vanishing probability as $n$ grows.

\begin{restatable}{lemma}{PMultiDraw}\label{lem:PMultiDraw}
Let $P_0$ denote the probability that the minimum-cost matching does not recover the true permutation. Then for $M = 2^n$, 
   \begin{align*}
       P_0 \geq  1-\left(1-p^{6}(1-p)\right)^{n-3}.
   \end{align*}
\end{restatable}

In the next section, we address Problem~\ref{prob1}.  

\section{Uniqueness of the $N$-permutation}
\label{sec:uniq_n}
The task of identifying the true permutation $\pi$, can be split into two steps. We can first identify the \textit{partitioning} $\{\cS_N((\bfx_i,\bfd_i)):i\in[M]\}$ and then for each \textit{partition} ($\cS_N((\bfx_i,\bfd_i))$) identify the \textit{label}, viz. the channel input ($\bfx_i$), where $i \in [M]$. Hence, given $R' $ and $\cA$, we are able to find the true permutation if and only if there exists only one valid partitioning and one valid labelling.

In Lemmas~\ref{L_th} and~\ref{N_th}, we determine the values $L_{\mathsf{Th}}$ and $N_{\mathsf{Th}}$, respectively, such that for all $L\geq L_\mathsf{Th}$ and $N\geq N_{\mathsf{Th}}$, we are able to find the true permutation with high probability. The result is formally stated in Theorem~\ref{thm:prob1}.

Before formally defining partitioning and labelling, we introduce some notations.

\begin{definition}
    For,  $\bfa,\bfb\in\{0,1\}^{n}$, let $\bfa',\bfb'$ be the channel outputs through the $\BDC(p)$ of $\bfa,\bfb$, respectively. We say that $\bfa'$ and $\bfb'$ are confusable, denoted by $\bfa' \cong \bfb'$, if $|{\SCS}(\bfa',\bfb')| \leq n$. Furthermore, let $\beta_{p}(\bfa,\bfb)$ denote the probability that $\bfa'$ and $\bfb'$ are confusable.  
\end{definition}

In Algorithm~\ref{Alg:Compute_Beta}, we describe how to compute $\beta_{p}(\bfa,\bfb)$. 

\begin{algorithm}[]
\caption{Total Probability of Confusable Events}
\label{Alg:Compute_Beta}
\begin{algorithmic}[1]
\Procedure{Compute Beta}{$p, \bfa,\bfb$}
\State $\beta_{p}(\bfa,\bfb) = 0, k_1 = 0, k_2 = 0$
\While{$k_1 \leq n$}
    \For{$\bfa' \in D_{k_1}(\bfa)$}
        \While{$k_2 \leq n$}
            \For{$\bfb' \in D_{k_2}(\bfb)$}
                % \State $\vert\SCS(\bfa',\bfb')\vert = (n-k_1) + (n-k_2) -\vert\LCS(\bfa',\bfb')\vert$
                \If{$\vert\SCS(\bfa',\bfb')\vert \leq n$}
                \State Let $k = k_1 + k_2$
                \State $\beta_{p}(\bfa,\bfb) \overset{+}{=} \omega_{\bfa'}(\bfa)\cdot\omega_{\bfb'}(\bfb) \cdot p^{k} \cdot (1-p)^{2n -k}$
                \EndIf
            \EndFor
            $k_2 \overset{+}{=} 1$
        \EndWhile
    \EndFor
    $k_1 \overset{+}{=} 1$
\EndWhile
\Return $\beta_{p}(\bfa,\bfb)$
\EndProcedure
\vspace{0.2cm}
\end{algorithmic}
\end{algorithm}

Let $(\bfx,\bfd), (\tx,\Tilde{\bfd})\in R$ then the read $(\bfy, \bfd')\in\cS_N((\bfx,\bfd))$ is said to be confusable if there exists some other read $(\Tilde{\bfy},\Tilde{\bfd'})\in R' /\{\cS_N((\bfy,\bfd'))\}$  such that $(\Tilde{\bfy},\Tilde{\bfd'})\in\cS_N((\Tilde{\bfx},\Tilde{\bfd}))$, $\bfy \cong \Tilde{\bfy}$ and $\bfd' \cong \Tilde{\bfd'}$, where $(\tx,\Tilde{\bfd}) \in R/\{(\bfx,\bfd)\}$. Let $\mathsf{R_{conf}}$ denote the multiset of such confusable reads.

In the next lemma, we describe a result on the length of the longest common subsequence of two uniformly chosen binary sequences, which would be crucial in analyzing the probability of a read being confusable.

\begin{lemma} \cite{L09}
\label{lem:Xkl}
    Let $X_{k,\ell}$ denote the length of the longest common subsequence of two uniformly chosen binary strings of length $k$ and $\ell$, respectively and $\lambda>0$, then 
    \begin{align}
        \mathbb{E}[X_{k,\ell}] \leq \frac{\gamma_2 (k+\ell)}{2} \leq \gamma_2 \max\{k,\ell\},
    \end{align}
    where $0.788 \leq \gamma_2 \leq 0.8263$. Further,
    \begin{align}
        P\left(\left|X_{k,\ell}-\mathbb{E}\left(X_{k,\ell}\right)\right| \geq \lambda\right) \leq 2 e^{-\frac{\lambda^{2}}{2(k+\ell
)}} \leq 2 e^{-\frac{\lambda^{2}}{4\max\{k,\ell\}}}.
    \end{align}
\end{lemma}

\begin{definition}
    For $c>0$, let $A^*_{c}$ denote the event that all noisy reads have length between $[(1-p)L-cL,(1-p)L+cL]$. 
\end{definition}

\begin{restatable}{lemma}{probAcL}\label{lem:probAcL}
The probability of the event $A^*_{c}$ is at least 
\begin{align*}
    P(A^*_{c}) \geq  \left(1-2e^{\left(-2c^{2}L\right)}\right)^{N2^n},
\end{align*}
where $c > 0$. 
\end{restatable}

Recall that $L = \Delta n$.

\begin{restatable}{corollary}{ProbAcLB}\label{cor:ProbAcLB}
For $c \geq \sqrt{\frac{2}{\Delta}}$, we have that  
$P(A^*_{c}) \geq 1 -\frac{2N}{2^n}$. 
\end{restatable}

In the next lemma, we calculate the probability of a read being confusable.
\begin{restatable}{lemma}{PFaulty}\label{lem:PFaulty}
    Let $c =\sqrt{\frac{2}{\Delta}}, \Delta > 2\left(\dfrac{(\gamma_2 + 2)}{(1-p)(2-\gamma_2)-1}\right)^2$ and $(\bfx,\bfd) \in R$. For $(\bfy,\bfd')\in\cS_N((\bfx,\bfd))$, we have that 
\begin{align*}
P(\mathbb{I}_{(\bfy,\bfd')\in \mathsf{R_{conf}}}) &< 1-\prod_{\bfx' \in \cA/{\bfx}}\left(1-\beta(\bfx,\bfx')2e^{-\theta L}\right)^{N} + \frac{2N}{2^n},
\end{align*}
where $\theta = \frac{\Big((2-\gamma_2)(1-p)-1 - \sqrt{\frac{2}{\Delta}}(2+\gamma_2)\Big)^2}{4}$. 
% where $\theta = \frac{\Big((2-\gamma_2)(1-p)-1 - \sqrt{\frac{2}{\Delta}}(2+\gamma_2)\Big)^2}{8(1-p+\sqrt{\frac{2}{\Delta}})}$.
\end{restatable}

The next corollary follows by observing that $\beta(\bfx,\bfx') < 1$ for all $\bfx,\bfx' \in \cA$. 
\begin{corollary}
Let $c =\sqrt{\frac{2}{\Delta}}, \Delta > 2\left(\dfrac{(\gamma_2 + 2)}{(1-p)(2-\gamma_2)-1}\right)^2$ and $(\bfx,\bfd) \in R$. For $(\bfy,\bfd')\in\cS_N((\bfx,\bfd))$, we have that 
\begin{align*}
P(\mathbb{I}_{(\bfy,\bfd')\in \mathsf{R_{conf}}}) < 1-\left(1-2e^{-\theta L}\right)^{N2^n} + \frac{2N}{2^n}, 
\end{align*}
where $\theta = \frac{\Big((2-\gamma_2)(1-p)-1 - \sqrt{\frac{2}{\Delta}}(2+\gamma_2)\Big)^2}{4}$.
\end{corollary}

\begin{definition}
     A \textbf{partitioning} $\cP=\{P_1, P_2, \ldots, P_M\}$ of $\cY$ is defined as the collection of disjoint submultisets of $\cY$, each of size $N$, such that for $i\in[M]$, for $(j,k)\in\binom{[N]}{2}$, $|\SCS(\bfy_j, \bfy_k)|\leq n \textrm{ and } |\SCS(\bfd'_j, \bfd'_k)|\leq L$, where $(\bfy_j, \bfd'_j), (\bfy_k, \bfd'_k) \in P_i$. 
\end{definition}
We will refer to $\cP^* \triangleq \{\cS_N((\bfx_i,\bfd_i)):i\in[M]\}$ as the \textit{true partitioning} of $R'$. Let $\mathbb{P}_{R'}$ denote the set of all possible partitionings of $R'$. Note that if $\vert\mathbb{P}_{R'}\vert = 1$ then $\mathbb{P}_{R'} = \{\cP^*\}$. Let $\cG' = (\cY, E')$, where $\cY = R'$. For $(\bfy,\bfd'),(\Tilde{\bfy},\Tilde{\bfd'}) \in\cY, ((\bfy,\bfd'),(\Tilde{\bfy},\Tilde{\bfd'}))\in E'$ if $(\bfy,\bfd') \cong(\Tilde{\bfy},\Tilde{\bfd'})$.
Note that a partitioning $\cP \in \mathbb{P}_{R'}$ corresponds to partitioning the graph $\cG'$ into $M$ cliques each of size $N$. 

\begin{proposition}
    $\vert\mathbb{P}_{R'}\vert=1$ if and only if there exists a unique partitioning of the graph $\cG'$ into $M$ cliques each of size $N$. 
\end{proposition}
  
In the next lemma, we derive a threshold on $L$ such that for $L\geq L_{\mathsf{Th}}, \mathbb{P}_{R'} = \{\cP^*\}$ with probability at least $1-\epsilon_1$.

\begin{restatable}{lemma}{BoundBeta}\label{L_th}
Let $n > \frac{1}{ln2}\left(\frac{N}{N-1}\right)ln\left(\frac{2N^2}{\epsilon_1^{\frac{1}{N}}}\right)$.  For 
\begin{align*}
    L \geq  \frac{2n}{\psi^2}\left(\phi + 2\sqrt{\frac{1}{n}ln\left(\frac{2N}{\epsilon_1^{\frac{1}{N}}}\right)}\right)^2 \triangleq L_{\mathsf{Th}}, 
\end{align*}
where $\phi = \gamma_2 + 2 + 2\sqrt{\ln{2}\left(\frac{N+2}{N}\right)}$ and $\psi = (1-p)(2-\gamma_2)-1$, we have that $\mathbb{P}_{R'} = \{\cP^*\}$ with probability at least $1-\epsilon_1$.
\end{restatable}

\begin{proof}    
Note that for every $\bfx\in\cA$, if there exists at least one $(\bfy,\bfd')\in\cS_N((\bfx,\bfd))$ such that 
 $(\bfy,\bfd')$  is not confusable, then the only valid partitioning is $\cP^*$. Let $\mathsf{X_{conf}}$ denote the set of left nodes with $\cS_N((\bfx,\bfd))\subset\mathsf{R_{conf}}$. 
From Markov Inequality, 
\begin{align*}
     P(\bfx \in \mathsf{X_{conf}}) 
     & = P(\mathbb{I}_{\cS_N((\bfx,\bfd))\subset \mathsf{R_{conf}}}\geq1)\\
     & \leq\mathbb{E}\left[\mathbb{I}_{\cS_N((\bfx,\bfd))\subset \mathsf{R_{conf}}}\right] 
     = \left(P(\mathbb{I}_{(\bfy,\bfd')\subset \mathsf{R_{conf}}} = 1)\right)^N.
\end{align*}
Therefore, from Lemma \ref{lem:PFaulty}, $P(\bfx \in \mathsf{X_{conf}})$ is at most 
\begin{align*}
   \left(1-\left(1-2e^{-\theta L}\right)^{N2^n} + \frac{2N}{2^n}\right)^N.
\end{align*}
From linearity of expectation, $\mathbb{E}\left[|\mathsf{X_{conf}}|\right]$ is at most
\begin{align*}
2^n \left(1-\left(1-2e^{-\theta L}\right)^{N2^n} + \frac{2N}{2^n}\right)^N.
\end{align*}
From Markov inequality, $P(|\mathsf{X_{conf}}|\geq 1)\leq \mathbb{E}\left[|\mathsf{X_{conf}}|\right]$. Hence,  $ P(|\mathsf{X_{conf}}|<1)$ is at least
\begin{align*}
   &= 1-2^n \left(1-\left(1-2e^{-\theta L}\right)^{N2^n} + \frac{2N}{2^n}\right)^N. 
\end{align*}
Further, using Bernoulli's Inequality, we get
\begin{align*}
    P(|\mathsf{X_{conf}}|<1) &> 1-2^n \left(2N2^ne^{-\theta L} + \frac{2N}{2^n}\right)^N.
\end{align*}
It can be verified that $P(|\mathsf{X_{conf}}|<1)\geq 1-\epsilon_1$ if $ \theta > \frac{1}{\Delta n}\left(\ln{(2N2^n)} - ln{\left(\left(\frac{\epsilon_1}{2^n}\right)^{\frac{1}{N}}-\frac{2N}{2^n}\right)}\right)$. \\

Since $n > \frac{1}{ln2}\left(\frac{N}{N-1}\right)ln\left(\frac{2N^2}{\epsilon_1^{\frac{1}{N}}}\right) \triangleq n_0$, we have that 
\begin{align*}
\left(\frac{\epsilon_1}{2^n}\right)^{\frac{1}{N}}-\frac{2N}{2^n} > \left(\frac{\epsilon_1}{2^{2n}}\right)^{\frac{1}{N}}.
\end{align*}
Let $\tau = \ln{2}\left(\frac{N+2}{N}\right)$ and $\gamma = \ln\left(\frac{2N}{\epsilon_1^{\frac{1}{N}}}\right)$.
Therefore, for $n >n_0$, if 
\begin{align*}
    \theta > \frac{\tau}{\Delta} + \frac{\gamma}{n\Delta}.
\end{align*} 
then we get $P(|\mathsf{X_{conf}}|<1)\geq 1-\epsilon_1$. Furthermore, if
\begin{align*}
\Delta > \left(\dfrac{(\gamma_2 + 2 + 2\sqrt{\tau} + 2\sqrt{\frac{\gamma}{n}})}{(1-p)(2-\gamma_2)-1}\right)^2   
\end{align*}
then $\theta > \dfrac{\tau}{\Delta} + \dfrac{\gamma}{n\Delta}$.

% \begin{align*}
%     1 - \epsilon_1 &< 1-2^n \left(1-\left(1-2e^{-\theta L}\right)^{N2^n} + \frac{2N}{2^n}\right)^N \\
%     \left(\frac{\epsilon_1}{2^{n}}\right)^{\frac{1}{N}} &>  1-\left(1-2e^{-\theta L}\right)^{N2^n} + \frac{2N}{2^n}\\
%     \left(1  + \frac{2N}{2^n} - \left(\frac{\epsilon_1}{2^{n}} \right)^{\frac{1}{N}}\right)^{\frac{1}{N2^n}} &< 1-2e^{-\theta L}\\
%     \left(1  + \frac{1}{N2^n}\left(\frac{2N}{2^n} - \left(\frac{\epsilon_1}{2^{n}} \right)^{\frac{1}{N}}\right)\right) &< 1-2e^{-\theta L}\\
%     % \ln\left(1 - \left(1 - \left(\frac{\epsilon_1}{2^{n+N}}\right)^{\frac{1}{N}}\right)^{\frac{1}{N2^n}}\right) &> \ln{2} -\theta L \\
%     % \ln\left(\frac{2}{\left(1 - \left(1 - \left(\frac{\epsilon_1}{2^{n+N}}\right)^{\frac{1}{N}}\right)^{\frac{1}{N2^n}}\right)}\right) &< \theta L
% \end{align*}

% We will operate in the region where the following is true
% \begin{align*}
%     \left(1 - \left(\frac{\epsilon_1}{2^{n+N}}\right)^{\frac{1}{N}}\right)^{\frac{1}{N2^n}} = 1 - \frac{1}{N2^n}\left(\frac{\epsilon_1}{2^{n+N}}\right)^{\frac{1}{N}}\\
%     \ln\left(\frac{2}{\left(1 - \left(1 - \left(\frac{\epsilon_1}{2^{n+N}}\right)^{\frac{1}{N}}\right)^{\frac{1}{N2^n}}\right)}\right) = \ln\left(\frac{4N}{\epsilon_1^{\frac{1}{N}}}\right) + n\left(\frac{N+1}{N}\right)\ln(2)\\
% \end{align*}

\end{proof}

\begin{definition}
    Given a partitioning $\cP=\{P_1,P_2, \ldots, P_M\}$, we define a \textbf{labelling}, denoted by $\cL$, as a length-$M$ vector of distinct addresses from $\cA$ such that $\cL[i] \in \{\bfx: \forall (\bfy,\bfd')\in P_i, P(\bfx|\bfy)>0\}$, where $\cL[i]$ denotes the $i$-th element of $\cL$, and $i\in[M]$.  
\end{definition}

We denote the set of all possible labellings for a given partitioning $\cP$ by $\mathbb{L}_{\cP,R'}$. Given the true partitioning $\cP^*$, we define the \textit{true labelling}, denoted by $\cL^*$, as the labelling in which for each partition $\cS_N((\bfx_i,\bfd_i))$, the assigned label is $\bfx_i$, where $i\in[M]$. Note that if $\cP\neq\cP^*$ then $\cL^*\notin\mathbb{L}_{\cP,R'}$. Further, if $|\mathbb{L}_{\cP^*,R'}| = 1$ then $\mathbb{L}_{\cP^*,R'} = \{\cL^*\}$. Let $\cG'' = (\cX,E'')$, where $\cX = \cA$. There is a directed edge $\bfx\rightarrow\tx$ if $\{\tx\} \in \{\bigcap_{(\bfy,\bfd')\in\cS_N((\bfx,\bfd))}E_{(\bfy,\bfd')}\}$. 

\begin{proposition}
$|\mathbb{L}_{\cP^*,R'}| = 1$ if and only if there are no directed cycles in $\cG''$.    
\end{proposition}

In the next lemma, we derive  a threshold on $N$ such that for $N\geq N_\mathsf{Th}, \mathbb{L}_{\cP^*,R'} = \{\cL^*\}$ with probability at least $1-\epsilon_2$. 

\begin{algorithm}[]
\caption{}
\label{Alg:Compute_Zeta}
\begin{algorithmic}[1]
\Procedure{Compute Zeta}{$p, \bfa,\bfb$}
\State $\zeta_{p}(\bfa,\bfb) = 0, k_1 = 0$
\While{$k_1 \leq n$}
    \For{$\bfa' \in D_{k_1}(\bfa)$}
            \If{$\bfa' \in D_{k_1}(\bfb) $}
            \State $\zeta_{p}(\bfa,\bfb) \overset{+}{=} \omega_{\bfa'}(\bfa)\cdot p^{k_1} \cdot (1-p)^{n - k_1} $
            \EndIf
    \EndFor
    $k_1 \overset{+}{=} 1$
\EndWhile
\Return $\zeta_{p}(\bfa,\bfb)$
\EndProcedure
\vspace{0.2cm}
\end{algorithmic}
\end{algorithm}

\begin{restatable}{lemma}{BoundOnN}\label{N_th}
     For $N\geq N_\mathsf{Th}$ 
     we have that $\mathbb{L}_{\cP^*,R'} = \{\cL^*\}$ with probability at least $1-\epsilon_2$, where $N_{\mathsf{Th}}$ is 
     \begin{align*}
    \underset{N\in\mathbb{Z_{+}}}{\argmin} \left(\sum_{\bfx\in \cX}\left(1-\prod_{\bfx'\in\cA /\{\bfx\}} \left(1-(\zeta(\bfx,\bfx'))^{N}\right)\right)  \leq \epsilon_2 \right)
     \end{align*}
\end{restatable}
\begin{proof}    
Let $\mathsf{X_{conf}}$ denote the set of nodes in $\cG''$ that have at least one outgoing edge. For $\bfx,\tx \in \cA$, the probability of $\bfx\rightarrow\tx$ is $(\zeta(\bfx,\bfx'))^{N}$. Therefore, the probability that $\bfx$ has no outgoing edges is $\prod_{\bfx'\in\cA /\{\bfx\}} \left(1-(\zeta(\bfx,\bfx'))^{N}\right)$. Hence, from linearity of expectation,
\begin{align*}
  \mathbb{E}\left[|\mathsf{X_{conf}}|\right] &= \sum_{\bfx\in \cX}\mathbb{E}[\mathbb{I}_{\{\bfx\in\mathsf{X_{conf}}\}}]\\
  &= \sum_{\bfx\in \cX}\left(1-\prod_{\bfx'\in\cA /\{\bfx\}} \left(1-(\zeta(\bfx,\bfx'))^{N}\right)\right).
\end{align*}
From Markov inequality, $P(|\mathsf{X_{conf}}|\geq 1)\leq \mathbb{E}\left[|\mathsf{X_{conf}}|\right]$. Hence, we get $P(|\mathsf{X_{conf}}|<1)$ is at least
\begin{align*}
    1-  \sum_{\bfx\in \cX}\left(1-\prod_{\bfx'\in\cA /\{\bfx\}} \left(1-(\zeta(\bfx,\bfx'))^{N}\right)\right). ~
\end{align*}
\end{proof}
\vspace{-0.2cm}
Thus, we define the region $\cR$ as $\cR\triangleq\{(\beta,N):\beta\geq\beta_\mathsf{Th}, N\geq N_{\mathsf{Th}}\}$. In the next theorem, we give a sufficient condition for the existence of a unique $N$-permutation. 

% \begin{figure}[h]
% \centering
% \includegraphics[width=7cm]{ISIT24/final_3.png}
% \caption{Plot of $N_{\mathsf{Th}}$ v/s $\epsilon_2$ for $n = 10, p = 0.15$.}
% \label{n_threshold}
% \end{figure}

\begin{restatable}{theorem}{TheOfUniquePermutation}\label{thm:prob1}
For $(L,N)\in\cR$, it is possible to identify the true permutation with probability at least $1-\epsilon$, if $\epsilon_1, \epsilon_2 < \frac{\epsilon}{2}$. 
\end{restatable}
\begin{proof}
    From Lemma~\ref{L_th} and~\ref{N_th}, it follows that for $L>L_\mathsf{Th}$ and $N>N_{\mathsf{Th}}$, $\mathbb{P}_{R'} = \{\cP^*\}$ with probability $(1-\epsilon_1)$ and $\mathbb{L}_{\cP^*,R'} = \{\cL^*\}$ with probability $(1-\epsilon_2)$, respectively. Hence, for $L>L_\mathsf{Th}$ and $N>N_{\mathsf{Th}}$, there exists only one valid permutation with probability $(1-\frac{\epsilon}{2})^2>(1-\epsilon)$. 
\end{proof}

\section{Permutation Recovery Algorithm}
\label{sec:perm_alg}
As previously mentioned, we split the task of identifying the true permutation into two steps. In the first step, we identify a partitioning via a clustering procedure.
We define the bipartite graph $\cG^{*} = ((\cA, \cP), \cE^{*})$, where $\cP$ is partitioning of $\cY$.  In the second step we find  a labelling for the partitioning $\cP$ using a minimum-cost algorithm (such as \cite{EK72,T71}).   

Let $\cN_{(\bfy,\bfd')}$ denote the two-hop neighborhood of $(\bfy,\bfd')$ in $\cG$. The clustering algorithm as described in Algorithm~\ref{alg:permRec}, iteratively selects the right node $(\bfy,\bfd')$ with the smallest two-hop neighborhood in $\cY$ and then performs $|\cN_{(\bfy,\bfd')}|$ data comparisons to identify the remaining $N-1$ copies. 

Let $\cP_\cG = \left(\cX \cup \cP, \cP_E\right)$ denote the bipartite matching identified by the minimum cost algorithm.  

\begin{algorithm}[]
\caption{Permutation Recovery Algorithm}
\label{alg:permRec}
\begin{algorithmic}[1]
\Procedure{Prune}{$(\Tilde{\bfy},\Tilde{\bfd'})$}
\State 
% $(\Tilde{\bfy},\Tilde{\bfd'})\longrightarrow\mathsf{Pruned}$
$(\widetilde{\bfy},\widetilde{\bfd'})\longrightarrow\mathsf{Pruned}$, $\cT = \{\{(\Tilde{\bfy},\Tilde{\bfd'})\}\}$
\For{$(\bfy,\bfd')\in  \cN_{(\Tilde{\bfy},\Tilde{\bfd'})}$}
    \If{$(\bfy,\bfd')\cong (\Tilde{\bfy},\Tilde{\bfd'})$}
    \State $(\bfy,\bfd') \longrightarrow \cT $
    \EndIf
\EndFor

\If{$|\cT| = N$}
\State Let $\cX^* = \bigcap_{(\bfy,\bfd')\in\cT}E_{(\bfy,\bfd')}$
\State
$W = \{w_{\bfx} \triangleq \gamma((\bfx,\cT)) : \bfx \in \cX^*\}$
\For{$(\bfy,\bfd')\in \cT$}    
    \State Remove $\{(\bfx, (\bfy,\bfd')):\bfx\notin \cX^*\}$ from $E$
    % \State $(\bfy,\bfd')\longrightarrow\mathsf{Pruned}$
    \For{$\bfx \in \cX^*$}
        \State
        $w_{\bfx} \overset{+}{=} \gamma(\bfx, (\bfy,\bfd'))$
    \EndFor
\EndFor

\State $\cT \longrightarrow\mathsf{\cP}$

\For{$\bfx \in \cX^*$}
    \State
    $\left(\bfx, \cT, \dfrac{w_{\bfx}}{N}\right)\longrightarrow\mathsf{\cE^*}$
    \State
    % $\bfx \longrightarrow\mathsf{\cP}$
\EndFor
    
\EndIf
\EndProcedure
\vspace{0.2cm}

\Procedure{Clustering Algorithm}{$\cP_\cG, \cG, \cG^*$}
\State$\mathsf{Pruned}=\{\{\}\}$
\While{$|\mathsf{Pruned}|<N2^n$}
\State $(\Tilde{\bfy},\Tilde{\bfd'}) = \argmin\{|\cN_{(\bfy,\bfd')}|:(\bfy,\bfd')\in \cY\}$
\State PRUNE $((\Tilde{\bfy},\Tilde{\bfd'}))$
\EndWhile
\State \Return $\mathsf{MCM}(\cP_\cG,\cG^*)$

\EndProcedure
\end{algorithmic}
\end{algorithm}

\begin{restatable}{proposition}{ProlgFindsPermutation}
\label{prop2}
For $(L,N)\in\cR$, Algorithm~\ref{alg:permRec} finds the true permutation with probability at least $1-\epsilon$, when $\epsilon_1, \epsilon_2 < \frac{\epsilon}{2}$. 
\end{restatable}
\begin{proof}
    For $L>L_{\mathsf{Th}}$, every address node has at least one channel output which is not confusable with probability at least $(1-\epsilon_1)$. Thus, the clustering algorithm identifies the true partitioning with probability at least $(1-\epsilon_1)$. For $N>N_{\mathsf{Th}}$, each left node has exactly one edge in $\cG^*$, thus there exists only one labelling, viz. the true labelling with probability at least $(1-\epsilon_2)$. Thus, the permutation recovery algorithm identifies the true permutation with probability at least $(1-\frac{\epsilon}{2})^2>(1-\epsilon)$. 
\end{proof}

   In the next lemma, we derive an upper bound on the expected number of comparisons performed by Algorithm~\ref{alg:permRec}.

\begin{lemma}\label{lem:expComp}
The expected number of comparisons performed by Algorithm~\ref{alg:permRec} is at most 
\begin{align*}
    \sum_{\bfx \in \cA}\sum_{\bfx'\in \cA} N\beta(\bfx,\bfx').
\end{align*}
\end{lemma}

\begin{proof}
    Note that for $(\bfy,\bfd'), (\Tilde{\bfy},\Tilde{\bfd'})\in\cY$,  $(\bfy,\bfd')\in\cN_{(\Tilde{\bfy},\Tilde{\bfd'})}$ if and only if $\bfy\cong\Tilde{\bfy}$. Therefore,
    \begin{align*}
        \mathbb{E}\left[|\cN_{(\bfy,\bfd')}|\right] &= N-1 + \sum_{\bfx'\in \cA/\{\bfx\}} ~\sum_{\Tilde{\bfy}\in\cS_N(\Tilde{\bfx})}\mathbb{I}_{\bfy \cong \Tilde{\bfy}}.
    \end{align*}
    Since, $P(\bfy \cong \Tilde{\bfy}) = \beta(\bfx,\tx)$, we get that
    \begin{align*}
       \mathbb{E}\left[|\cN_{(\bfy,\bfd')}|\right]  \leq \sum_{\bfx'\in \cA} N\beta(\bfx,\bfx').
    \end{align*}
    Hence, the expected number of comparisons performed by the algorithm is at most $\sum_{\bfx \in \cA}\sum_{\bfx'\in \cA} N\beta(\bfx,\bfx').$
\end{proof}

Since $\beta(\bfx,\bfx') < 1$ for all $\bfx,\bfx' \in \cA$, the expected number of data comparisons performed by Algorithm \ref{alg:permRec} is only a $\kappa_{n,\cA}$-fraction of data comparisons required by clustering based approaches, where $\kappa_{n,\cA} = \dfrac{\sum_{\bfx \in \cA}\sum_{\bfx'\in \cA} \beta(\bfx,\bfx')}{M^2} $.

\newpage
% Generated by IEEEtran.bst, version: 1.13 (2008/09/30)

\clearpage
\appendix
\PMultiDraw*
\begin{proof} Let $\bfx \in \cA, (\bfy,\bfd') \in R'$. For a deletion channel with deletion probability $p$, the likelihood probability is 
\begin{align*}
P(\bfy \mid \bfx)=\omega_{\bfy}(\bfx) \cdot p^{n-\vert\bfy\vert}(1-p)^{\vert\bfy\vert}.
\end{align*}
Let $\bfx,\tx \in \cA$, $(\bfy, \bfd')\in\cS_N((\bfx,\bfd)), (\Tilde{\bfy},\Tilde{\bfd'})\in\cS_N((\Tilde{\bfx},\Tilde{\bfd}))$. Consider the following event denoted by $\cB$:
\begin{align*}
\gamma(\tx, (\bfy,\bfd'))<\gamma(\bfx, (\bfy,\bfd')) &\wedge \gamma(\bfx,(\Tilde{\bfy},\Tilde{\bfd'}))<\gamma(\tx, (\Tilde{\bfy},\Tilde{\bfd'})). 
\end{align*}
If $\cB$ is true then the minimum-cost algorithm will necessarily assign $(\bfy,\bfd'), (\Tilde{\bfy},\Tilde{\bfd'})$ wrongfully. However, not necessarily $(\Tilde{\bfy},\Tilde{\bfd'})$ to $\bfx$ and $(\bfy,\bfd')$ to $\tx$. Further,
\begin{align*}
P(\bfy \mid \bfx)&=\omega_{\bfy}(\bfx) \cdot p^{n-\vert\bfy\vert}(1-p)^{\vert\bfy\vert},\\
P(\bfy \mid \tx)&=\omega_{\bfy}(\tx) \cdot p^{n-\vert\bfy\vert}(1-p)^{\vert\bfy\vert}.
\end{align*}
Therefore, $\gamma(\tx, (\bfy,\bfd')) < \gamma(\bfx, (\bfy,\bfd'))$ if and only if $\omega_{\bfy}(\bfx) < \omega_{\bfx}(\tx)$. 

Let $m \leq n-4, k \leq n-3$. Let $\bfx =0^{m} 10^{n-m-3} 11$ and $\tx =0^{m} 10^{n-m-4} 100$. Furthrer 
Let $\bfy_{k}=0^{k}, \Tilde{\bfy}_{k}=0^{k-1} 1$. For each $m, n, k$, it holds that
$$
\begin{aligned}
& \omega_{\bfy_{k}}(\bfx) =\binom{n-3}{k}<\binom{n-2}{k}=\omega_{\bfy_{k}}(\tx).\\
& \omega_{\Tilde{\bfy}_{k}}(\bfx)=\left\{\begin{array}{cc}
2\binom{n-3}{k-1} & k-1>m, \\
2\binom{n-3}{k-1}+2\binom{m}{k-1} & k-1 \leq m.
\end{array}\right. \\
& \omega_{\Tilde{\bfy}_{k}}(\tx)=\left\{\begin{array}{cc}
\binom{n-4}{k-1} & k-1>m, \\
\binom{n-4}{k-1}+2\binom{m}{k-1} & k-1 \leq m.
\end{array}\right.
\end{aligned}
$$
Therefore,
\begin{align*}
    \omega_{\Tilde{\bfy}_{k}}(\bfx) >\omega_{\Tilde{\bfy}_{k}}(\tx).
\end{align*}
Note that the probability that an output of $\bfx$
is a subsequence of $\tilde{\bfy}_{n-3}$ is at least $p^3(1-p)$. Similarly, the probability that an output of $\tx$ is a subsequence of $\bfy_{n-3}$ is $p^3$. 
Therefore, for any $k\leq n-3$, when $\bfy_{k}$ is a noisy copy of $\bfx$ and $\Tilde{\bfy}_{k'}$ is a noisy copy of $\tx$, the min cost algorithm wrongfully assigns $(\bfy_{k},\bfd')$ and $(\Tilde{\bfy}_{k'},\Tilde{\bfd'})$ with probability at least 
\begin{align*}
    P_{0} & \geq 1-\left(1-p^{6}(1-p)\right)^{n-3}.
\end{align*}
\end{proof}

\probAcL*
\begin{proof}
Let $d' \in \{0,1\}^{*}$, then it can be verified that $\mathbb{E}(|d^{\prime}|) = (1-p)L$. Further, using Hoeffding's inequality, we have that
\begin{align*}
    P\left(\big\vert \left|d^{\prime}\right| -(1-p)L\big\vert \geq cL \right) &\leq 2e^{\left(-\frac{2(cL)^{2}}{L}\right)},
\end{align*}
where $c>0$. Therefore, the probability of the event $A^*_{c}$ is at least 
\begin{align*}
      \left(1-2e^{\left(-2c^{2}L\right)}\right)^{N2^n}.
\end{align*}
\end{proof}

\ProbAcLB*
\begin{proof} 
Let $\epsilon_n > 0$. Using Bernoulli's inequality we have that 
\begin{align*}
    % 1-\frac{1}{N2^n}\epsilon_n &> (1-\epsilon_n)^{\frac{1}{N2^n}}\\
    % \ln{\frac{2}{1-\left(1-\frac{1}{N2^n}\epsilon_n\right)}} &> \ln{\frac{2}{1-(1-\epsilon_n)^{\frac{1}{N2^n}}}}\\
    \sqrt{\dfrac{\left(\ln{\frac{2}{1-(1-\epsilon_n)^{\frac{1}{N2^n}}}}\right)}{2}} &<
    \sqrt{\dfrac{\ln{\frac{2}{1-\left(1-\frac{1}{N2^n}\epsilon_n\right)}}}{2}}\\
    % \sqrt{\dfrac{\left(\ln{\frac{2}{1-(1-\epsilon_n)^{\frac{1}{N2^n}}}}\right)}{2}} 
    &= \sqrt{\frac{1}{2}\ln{\frac{2N}{\epsilon_n}} + \frac{\ln{2}}{2}n}\\
    % \sqrt{\dfrac{\left(\ln{\frac{2}{1-(1-\epsilon_n)^{\frac{1}{N2^n}}}}\right)}{2}} 
    &< \sqrt{\frac{1}{2}\ln{\frac{2N}{\epsilon_n}}} + \sqrt{\frac{\ln{2}}{2}}\sqrt{n} \triangleq c_0 \sqrt{L}.
\end{align*}
Hence, if $c \geq \frac{c_0}{\sqrt{\ln2}} > c_0$ then  from Lemma \ref{lem:probAcL}, we get $P(A^*_{c}) \geq 1 - \epsilon_n$.  When $\epsilon_n = \frac{2N}{2^n} $, we get $\frac{1}{\sqrt{\ln2}}c_0 = \sqrt{\frac{2}{\Delta}}$ as desired. 

\end{proof}

\PFaulty*
\begin{proof}
Let $(\tx,\tilde{\bfd})\in R$ and $(\Tilde{\bfy},\Tilde{{\bfd'}}) \in \cS_N((\tx,\tilde{\bfd}))$. From  Lemma~\ref{lem:Xkl}, we have that 
$\mathbb{E}\left(\left|\LCS\left(d^{\prime}, \tilde{d'}\right)\right|\right) \leq \gamma_{2} \max\left\{|d'|,|\tilde{d'}|\right\}$. Given $A^*_{c}$, we have that
\begin{align*}
\mathbb{E}\left(\left|\LCS\left(d^{\prime}, \tilde{d'}\right)\right|~\big\vert~ A^*_{c}\right) \leq \gamma_{2}\left((1-p)L+cL\right).
\end{align*}
Since $\left|\SCS\left(d^{\prime}, \tilde{d'}\right)\right|=\left|d^{\prime}\right|+\left|\tilde{d'}\right|-\left|\LCS\left(d^{\prime}, \tilde{d'}\right)\right|$, we have
 \begin{align*}
\mathbb{E}\left(\left|\SCS\left(d^{\prime}, \tilde{d'}\right)\right|\right) \geq |d'| + |\tilde{d'}| - \gamma_{2} \max \left\{|d'|,|\tilde{d'}|\right\}.
 \end{align*}

We now analyze the probability that  $\left|\SCS\left(d^{\prime}, \tilde{d'}\right)\right| < L$ given $A^*_{c}$. For brevity, we let $Z = \left|\LCS\left(d^{\prime}, \tilde{d'}\right)\right|$.
\begin{align*}
    & P\Big(\left|\SCS\left(d^{\prime}, \tilde{d'}\right)\right| < L ~\Bigg\vert~A^*_{c}\Big)\\
    &\leq P\Big(2\left((1-p)L-cL\right) - Z < L~\Bigg\vert~A^*_{c}\Big)\\
    &= P\Big(\mathbb{E}\left[Z\right] - Z <  \mathbb{E}\left[Z\right] -2\left((1-p)L-cL\right) +L ~\Bigg\vert~A^*_{c} \Big)\\
    &\leq  P\Big(\mathbb{E}\left[Z\right]-Z <  -\rho ~\Bigg\vert~A^*_{c} \Big)\\
    &=  P\Big(Z -\mathbb{E}\left[Z\right] >  \rho~\Bigg\vert~A^*_{c} \Big)\\
    &\leq  P\Big(\left\vert Z -\mathbb{E}\left[Z\right] \right\vert > \rho ~\Bigg\vert~A^*_{c} \Big),
\end{align*}
where $\rho = L((2-\gamma_2)(1-p)-1) - cL(2+\gamma_2)$. Since $\Delta > 2\left(\dfrac{(\gamma_2 + 2)}{(1-p)(2-\gamma_2)-1}\right)^2$ and $c = \sqrt{\frac{2}{\Delta}}$ we have $\rho > 0$. Therefore, from Lemma \ref{lem:Xkl}, we get  
\begin{align*}
    P\Big(\left|\SCS\left(d^{\prime}, \tilde{d'}\right)\right| < L ~\Bigg\vert~A^*_{c}\Big) \leq 2e^{-\dfrac{\rho^2}{4\max\{|d'|,|\tilde{d'}|\}}}
\end{align*}
% \sh{slight improvement is possible, to be done later}
Since $\max\{|d'|,|\tilde{d'}|\} < L$ and substituting $c = \sqrt{\frac{2}{\Delta}}$, we get 
\begin{align*}
     P\Big(\left|\SCS\left(d^{\prime}, \tilde{d'}\right)\right| < L ~\Bigg\vert~A^*_{c}\Big) &\leq 2e^{-\theta L},
\end{align*}
where $\theta = \frac{\Big((2-\gamma_2)(1-p)-1 - \sqrt{\frac{2}{\Delta}}(2+\gamma_2)\Big)^2}{4}$. \\

Since $P(\bfy \cong \Tilde{\bfy}) = \beta(\bfx,\tx)$, we have that
\begin{align*}   P\left((\Tilde{\bfy},\Tilde{{\bfd'}}) \cong (\bfy,\bfd') ~\Bigg\vert~A^*_{c} \right)  \leq \beta(\bfx,\tx)2e^{-\theta L}.
\end{align*}
Hence, we get $P\left(\mathbb{I}_{(\bfy,\bfd')\in \mathsf{R_{conf}}}~\Bigg\vert~A^*_{c} \right)$ as 
\begin{align*}
    &= 1 -\prod_{(\Tilde{\bfy},\Tilde{\bfd'}) \in R'/\cS_N((\bfx,\bfd))} 1-P \left((\bfy,\bfd') \cong(\Tilde{\bfy},\Tilde{\bfd'})~\Bigg\vert~A^*_{c} \right)\\
    &\leq 1-\prod_{\bfx' \in \cA/{\bfx}}\left(1-\beta(\bfx,\bfx')2e^{-\theta L}\right)^{N}.
\end{align*}
By law of total probability,
\begin{align*}
    P(\mathbb{I}_{(\bfy,\bfd')\in \mathsf{R_{conf}}})  &= P\left(\mathbb{I}_{(\bfy,\bfd')\in \mathsf{R_{conf}}}~\Bigg\vert~A^*_{c}\right)P(A^*_{c}) \\
    &+ P\left(\mathbb{I}_{(\bfy,\bfd')\in \mathsf{R_{conf}}}~\Bigg\vert~\overline{A^*}_{c}\right)P(\overline{A^*}_{c})\\
    &\leq P\left(\mathbb{I}_{(\bfy,\bfd')\in \mathsf{R_{conf}}}~\Bigg\vert~A^*_{c}\right) + (1-P(A^*_{c})).
\end{align*}
Since $c = \sqrt{\frac{2}{\Delta}}$, it follows from Corollary~\ref{cor:ProbAcLB} that 
\begin{align*}
P(\mathbb{I}_{(\bfy,\bfd')\in \mathsf{R_{conf}}}) < 1-\prod_{\bfx' \in \cA/{\bfx}}\left(1-\beta(\bfx,\bfx')2e^{-\theta L}\right)^{N} + \frac{2N}{2^n}.
\end{align*}
Since $\beta(\bfx,\bfx') < 1$, we get 
\begin{align*}
    P(\mathbb{I}_{(\bfy,\bfd')\in \mathsf{R_{conf}}}) < 1-\left(1-2e^{-\theta L}\right)^{N2^n} + \frac{2N}{2^n}.
\end{align*}
\end{proof}

\newpage
\begin{thebibliography}{99}
\providecommand{\url}[1]{#1}
\csname url@samestyle\endcsname
\providecommand{\newblock}{\relax}
\providecommand{\bibinfo}[2]{#2}
\providecommand{\BIBentrySTDinterwordspacing}{\spaceskip=0pt\relax}
\providecommand{\BIBentryALTinterwordstretchfactor}{4}
\providecommand{\BIBentryALTinterwordspacing}{\spaceskip=\fontdimen2\font plus
\BIBentryALTinterwordstretchfactor\fontdimen3\font minus
  \fontdimen4\font\relax}
\providecommand{\BIBforeignlanguage}[2]{{%
\expandafter\ifx\csname l@#1\endcsname\relax
\typeout{** WARNING: IEEEtran.bst: No hyphenation pattern has been}%
\typeout{** loaded for the language `#1'. Using the pattern for}%
\typeout{** the default language instead.}%
\else
\language=\csname l@#1\endcsname
\fi
#2}}
\providecommand{\BIBdecl}{\relax}
\BIBdecl

\bibitem{SBKY23}
S. Singhvi, A. Boruchovsky, H. M. Kiah and E. Yaakobi,  ``Data-Driven Bee Identification for DNA Strands," 2023 IEEE International Symposium on Information Theory (ISIT), Taipei, Taiwan, 2023, pp. 797-802.  


\bibitem{CKVY23}
J. Chrisnata, H. M. Kiah, A. Vardy and E. Yaakobi,  ``Bee Identification Problem for DNA Strands," \emph{IEEE Journal on Selected Areas in Information Theory}, vol. 4, pp. 190-204, 2023.

\bibitem{Church.etal:2012}
G.~M. Church, Y.~Gao, and S.~Kosuri.
``Next-generation digital information storage in {DNA},''
{\em Science}, vol. 337, no. 6102, pp. 1628--1628, 2012.

\bibitem{Goldman.etal:2013}
N.~Goldman, P.~Bertone, S.~Chen, C.~Dessimoz, E.~M. LeProust, B.~Sipos, and  E.~Birney.
``Towards practical, high-capacity, low-maintenance information storage in synthesized {DNA},''
{\em Nature}, vol. 494, no. 7435, pp. 77--80, 2013.


\bibitem{Kiah2021}
 H. M. Kiah,  A. Vardy, and  H. Yao, 
 ``Efficient bee identification,"  
 \textit{IEEE International Symposium on Information Theory (ISIT)}, pp. 1943--1948, July, 2021.

\bibitem{KVY2022}
H. M. Kiah, A.  Vardy, and H. Yao, 
``Efficient algorithms for the bee-identification problem," \textit{arXiv preprint} arXiv:2212.09952, 2022.

\bibitem{LSWY18}
A. Lenz, P. H. Siegel, A. Wachter-Zeh and E. Yaakobi,
``Coding over sets for DNA storage,"
\emph{IEEE Transactions on Information Theory}, vol. 66, no. 4, pp. 2331--2351, April 2020.


\bibitem{Organick}
L. Organick, S. Ang, Y.J. Chen, R. Lopez, S.Yekhanin, K. Makarychev, M. Racz, G. Kamath, P. Gopalan, B. Nguyen, C. Takahashi, S. Newman, H. Y. Parker, C. Rashtchian, K. Stewart, G. Gupta, R. Carlson, J. Mulligan, D. Carmean, G. Seelig, L. Ceze, and K. Strauss, 
``Random access in largescale DNA data storage," \textit{Nature Biotechnology}, vol. 36, no. 3, pp 242--248, 2018.

\bibitem{clustering}
C. Rashtchian, K. Makarychev, M. Racz, S. Ang, D. Jevdjic, S. Yekhanin, L. Ceze, and K. Strauss, 
``Clustering billions of reads for DNA data storage," \textit{Advances in Neural Information Processing Systems}, vol. 30, 2017.


\bibitem{Shomorony.2022}
I. Shomorony, and R. Heckel, 
``Information-theoretic foundations of DNA data storage,''
Foundations and Trends\textregistered in Communications and Information Theory, 19(1), 1--106, 2022

\bibitem{TTV2019}
A. Tandon , V.Y.F. Tan, and L.R. Varshney, 
``The bee-identification problem: Bounds on the error exponent,"
\textit{IEEE Transactions on Communications}, vol. 67, issue no.11, pp. 7405--7416, November, 2019.


\bibitem{Yazdi.etal:2015b}
S.~Yazdi, H.~M. Kiah, E.~R. Garcia, J.~Ma, H.~Zhao, and O.~Milenkovic, 
``{DNA}-based storage: Trends and methods," 
{\em IEEE Trans. Molecular, Biological, Multi-Scale Commun.},  vol. 1, no. 3, pp. 230--248, 2015.

\bibitem{L09}
G. S. Lueker,  `Improved bounds on the average length of longest common subsequences," Journal of the ACM (JACM), vol. 56, p. 1-38, 2009.

\bibitem{EK72}
J. Edmonds and R. M. Karp, ``Theoretical improvements in algorithmic efficiency for network flow problems,” J. ACM, vol. 19, no. 2,
pp. 248-264, 1972.

\bibitem{T71}
N. Tomizawa,``On some techniques useful for solution of transportation
network problems,” Networks, vol. 1, no. 2, pp. 173-194, 1971.

\end{thebibliography}
\end{document}